\documentclass[aps,pre,reprint,longbibliography,nobalancelastpage]{revtex4-2}

\usepackage{amsmath,amssymb,amsthm,amsbsy}
\usepackage{graphicx}
\usepackage{xcolor}
\usepackage{mathrsfs}
\usepackage{stmaryrd}
\usepackage{hyperref}
\usepackage{enumitem}
\usepackage{accents}
\usepackage{tikz}
\usetikzlibrary{tikzmark}

\usepackage{newtxtext}
\usepackage[slantedGreek]{newtxmath}
\renewcommand{\Delta}{\upDelta}

\usepackage[compact]{titlesec}

\definecolor{linkcolor}{HTML}{223096}
\hypersetup{colorlinks,allcolors=linkcolor}
\bibpunct{\textcolor{linkcolor}{[}}{\textcolor{linkcolor}{]}}{\textcolor{linkcolor}{,}}{n}{}{;}
\renewcommand{\eqref}[1]{\hyperref[#1]{(\ref*{#1})}}

\usetikzlibrary{decorations.pathreplacing,arrows}

\renewcommand{\vec}[1]{\boldsymbol{#1}}

\newcommand{\ic}{\mathrm{i}}

\DeclareMathAlphabet{\mathbfsf}{\encodingdefault}{\sfdefault}{b}{n}
\DeclareMathAlphabet{\mathbfsfit}{\encodingdefault}{\sfdefault}{b}{it}
\newcommand{\tens}[1]{\mathbfsfit{#1}}
\newcommand{\tzero}{\mathbfsf{0}}

\newcommand{\figref}[2]{[Fig.~\hyperref[#1]{\ref*{#1}(#2)}]}
\newcommand{\figrefi}[2]{[Fig.~\hyperref[#1]{\ref*{#1}(#2)}, inset]}
\newcommand{\textfigref}[2]{Fig.~\hyperref[#1]{\ref*{#1}(#2)}}
\newcommand{\textfigureref}[2]{Figure~\hyperref[#1]{\ref*{#1}(#2)}}
\newcommand{\wholefigref}[1]{(Fig.~\ref{#1})}

\newcommand{\textwholefigref}[1]{Fig.~\ref{#1}}

\newcommand{\figrefp}[2]{\hyperref[#1]{\ref*{#1}(#2)}}

\renewcommand{\geq}{\geqslant}
\DeclareMathOperator{\tr}{tr}

\DeclareMathAlphabet{\mathcal}{OMS}{cmsy}{m}{n}

\newcommand{\citeappendix}[1]{Appendix~\ref{#1}}

\newtheoremstyle{myremark}%
  {3pt}{3pt}
  {\normalfont}
  {}
  {\itshape}
  {.}
  { }
  {%
    \thmname{#1}%
    \thmnumber{\itshape\ #2}%
    \thmnote{ (#3)}%
  }

\theoremstyle{myremark}
\newtheorem{prop}{Proposition}
\newtheorem{lemma}{Lemma}
\newtheorem{obs}{Observation}

\begin{document}

\title{Travelling waves of invasion in microbial communities with phenotypic switching}
\author{Diego \surname{Manso Anda}}
\affiliation{Max Planck Institute for the Physics of Complex Systems, N\"othnitzer Stra\ss e 38, 01187 Dresden, Germany}
\affiliation{\smash{Max Planck Institute of Molecular Cell Biology and Genetics, Pfotenhauerstra\ss e 108, 01307 Dresden, Germany}}
\affiliation{Center for Systems Biology Dresden, Pfotenhauerstra\ss e 108, 01307 Dresden, Germany}
\date{\today}%
\author{Pierre A. Haas}
\email[Contact author: ]{haas@pks.mpg.de}
\affiliation{Max Planck Institute for the Physics of Complex Systems, N\"othnitzer Stra\ss e 38, 01187 Dresden, Germany}
\affiliation{\smash{Max Planck Institute of Molecular Cell Biology and Genetics, Pfotenhauerstra\ss e 108, 01307 Dresden, Germany}}
\affiliation{Center for Systems Biology Dresden, Pfotenhauerstra\ss e 108, 01307 Dresden, Germany}
\date{\today}%
\begin{abstract}
Complex microbial habitats see the spatial competition of different clonal bacterial populations that switch between different phenotypes. Here, we determine the effect of this subpopulation structure on the invasion of one species by another in a minimal model of two competing species: one species switches, both stochastically and in response to its competitor, to a persister phenotype resilient to competition. Surprisingly, our combined analytical and numerical results show that this phenotypic switching has no effect on the speed of the travelling wave by which the competitors invade the first population. Conversely, we discover that phenotypic switching can speed up the wave by which this population invades their competitors. Our results thus suggest, counterintuitively, that bacterial persistence can be an offensive, rather than defensive ecological strategy.
\end{abstract}

\maketitle

\renewcommand{\floatpagefraction}{.999}

\section{\uppercase{Introduction}}
In microbial communities, clonal bacterial populations divide into different phenotypes~\cite{smits06,avery06,dubnau06}. Despite the known ecological importance of this subpopulation structure~\cite{bolnick11,forsman16,turcotte16}, physical models of its effect on species coexistence in well-mixed populations have only emerged more recently~\cite{kussell05b,maynard19,haas20,haas22,holdrige22}. Abundant phenotypic variation is expected to be destabilising in general~\cite{haas20} from May's stability paradigm~\cite{may72,allesina15}, since it increases the effective number of species in the community. Still, different mechanisms by which phenotypic variation can increase diversity and stability have been identified~\cite{maynard19,haas20,holdrige22}. For example, phenotypic switching to a rare phenotype, such as the bacterial persister phenotype~\cite{balaban04,maisonneuve14,harms16,radzikowski17}, can stabilise coexistence~\cite{haas20,haas22}.

Complex microbial habitats are not however well-mixed. Rather, one expects spatial competition dynamics, with some bacterial species invading other species. In the simplest, one-dimensional case, with bacterial motion represented as diffusion, these dynamics can take the form of travelling waves: indeed, the Lotka--Volterra equations of population dynamics~\cite{* [] [{, Chap.~3, pp.~79--118 and Chap.~13, pp.~437--483.}] murray,hofbauer} with diffusion admit travelling-wave solutions \cite{dunbar84,lewis02,fei03,fang09,wang11} that are higher-dimensional analogues of the travelling waves of the Fisher equation~\cite{murray,vansarloos}. These travelling waves in the Lotka--Volterra system remain a very active area of research in applied mathematics~\cite{alhasanat19,tang22,wang22,ma24,cao25,chen25}.

\begin{figure}[b]
\centering\includegraphics{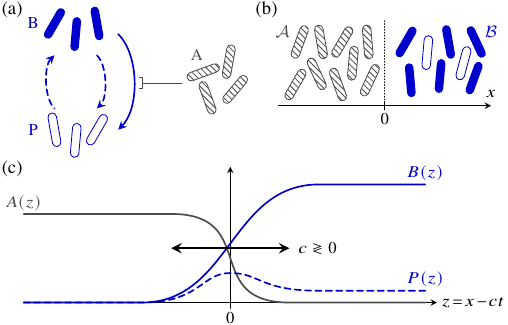}
\caption{Minimal model of the spatial competition of two species with phenotypic variation. (a)~A bacterial species has two phenotypes, a growing phenotype B and a persister phenotype P, between which it switches stochastically (dashed arrows). Additionally, the competitor species A induces responsive switching (solid arrow) from B to~P. (b)~Setup of the one-dimensional competition between these two species: at time $t=0$, the first species, in an equilibrium $\mathcal{B}$ of B and P, is in $x>0$, while the competitors, in an equilibrium $\mathcal{A}$ of A only, are in $x<0$. (c)~Dynamics of the spatial competition: after a transient, one equilibrium invades the other with a travelling wave of velocity $c$. The abundances $B(z),P(z),A(z)$ of B, P, A are fixed in the travelling-wave coordinate $z=x-ct$.}\label{fig1}
\end{figure}

Here, we ask: how do phenotypic variation and switching affect these travelling waves of invasion? To answer this question, we develop a minimal model of the spatial competition of two species~\wholefigref{fig1}: One species has two phenotypes, B and~P, between which it switches stochastically~\figref{fig1}{a}. Moreover, the competitor species A causes responsive phenotypic switching from B to P. These dynamics allow two particular types of homogeneous equilibria: equilibria $\mathcal{B}$ without competitors and equilibria $\mathcal{A}$ of competitors only. (We will not be interested in coexistence equilibria of both species.) We analyse the spatial dynamics resulting from these homogeneous equilibria meeting at $x=0$~\figref{fig1}{b}. After a transient, a travelling wave forms~\figref{fig1}{c}, with one of $\mathcal{A},\mathcal{B}$ invading the other at velocity $c$. (Again, we will not be interested in the possibility of a coexistence state invading both $\mathcal{A},\mathcal{B}$.) 

Combining analytical and numerical calculations, we determine how phenotypic switching affects the wave velocity~$c$. Contrary to the expectation that a hump of a phenotype resilient to the competitors should slow down the wave by which $\mathcal{A}$ invades $\mathcal{B}$~[$c>0$, \textfigref{fig1}{c}], we reveal that phenotypic switching does not affect $c>0$. Conversely, we find that phenotypic switching can speed up invasion of $\mathcal{A}$ by $\mathcal{B}$ [$c<0$, \textfigref{fig1}{c}]. In this way, our work suggests that phenotypic switching to a phenotype resilient to competition is, surprisingly, not a defensive, but rather an offensive ecological strategy.

\section{\uppercase{Results}}
Let $B(x,t),P(x,t),A(x,t)$ denote the abundances of populations B, P, A, where $t$ is time and $x$ is the (single) spatial coordinate. We model the dynamics of these abundances by the reaction-diffusion system
\begin{subequations}\label{eq:lv}
\begin{align}
\dfrac{\partial B}{\partial t}&=B(1 - \alpha A - B - \kappa P ) - \beta AB - \gamma B + \delta P + \frac{\partial^2 B}{\partial x^2}, \label{eq:dB}\\
\dfrac{\partial P}{\partial t}&=P(\mu - \xi A - \varpi B - \varsigma P ) + \beta AB + \gamma B - \delta P + d\frac{\partial^2 P}{\partial x^2},\label{eq:dP}\\
\dfrac{\partial A}{\partial t}&= A(\zeta - \eta A - \vartheta B - \iota P ) + D\frac{\partial^2 A}{\partial x^2},
\end{align}
\end{subequations}
in which $\alpha,\beta,\gamma,\delta,\zeta,\eta,\vartheta,\iota,\kappa,\mu,\xi,\varpi,\varsigma$ and $d,D$ are non-negative parameters. The non-spatial part of this model is the minimal model of phenotypic switching of two species introduced in Eqs.~(4) of Ref.~\cite{haas22}. Thus $\alpha,\zeta,\eta,\vartheta,\iota,\kappa,\mu,\xi,\varpi,\varsigma$ are Lotka--Volterra parameters representing the growth rates of and competition between the different populations, $\beta$ is the rate of responsive switching from B to P, and $\gamma,\delta$ are the deterministic rates of stochastic switching between B and P. The diffusive spatial terms add bacterial motility to this model, with diffusivities $d,D$.

Equations~\eqref{eq:lv} represent non-dimensional dynamics in which population abundances, time, and space have been rescaled so that $B$ has unit growth rate, carrying capacity, and diffusivity. As noted in Ref.~\cite{haas22}, too, the same dimensional rescalings must be chosen for $B$ and $P$ lest the switching terms in Eqs.~\eqref{eq:dB} and~\eqref{eq:dP} become unbalanced; choosing a different rescaling for~$A$ would allow setting, e.g., $\eta=1$, but would make comparing competition strengths more difficult.

While Eqs.~\eqref{eq:lv} provide a general model of phenotypic variation, it will turn out to be useful to focus at times on a particular biological realisation of the model, viz., bacterial persistence~\cite{balaban04,maisonneuve14,harms16,radzikowski17}. Thus B will represent the ``normal'' growing phenotype, while P will present the rare, slowly growing, and slowly moving persister phenotype resilient to the competitor~A. This corresponds to the scalings~\cite{haas22}
\begin{align}
&\alpha,\delta,\zeta,\eta,\vartheta,D=O(1),&&\beta,\gamma,\iota,\kappa,\mu,\xi,\varpi,\varsigma,d=O(\varepsilon),\label{eq:scalings}
\end{align}
in which $\varepsilon$ is a small parameter representing the relative rareness of persisters.

We conclude the setup of the model by analysing its homogeneous steady states. In particular, Eqs.~\eqref{eq:lv} have steady states
\begin{align}
&\mathcal{A}=(0,0,A_\ast),&&\mathcal{B}=(B_\ast,P_\ast,0),\label{eq:eqs}
\end{align}
of A only and B, P only, respectively, where $A_\ast=\zeta/\eta$ and $B_\ast,P_\ast$ are determined by
\begin{subequations}\label{eq:BPeq}
\begin{align}
B_\ast(1 - B_\ast - \kappa P_\ast) - \gamma B_\ast + \delta P_\ast&=0,\\
P_\ast(\mu-\varpi B_\ast-\varsigma P_\ast)+\gamma B_\ast-\delta P_\ast&=0.
\end{align}
\end{subequations}
Equations~\eqref{eq:BPeq} cannot be solved usefully in closed form, but it is easy to see that, if $\varepsilon\ll1$, then they have the unique non-trivial solution $B_\ast=1+O\bigl(\varepsilon^2\bigr)$, $P_\ast=\gamma/\delta+O\bigl(\varepsilon^2\bigr)$, which is real and positive. More generally, we prove the following in \citeappendix{appA}:
\begin{prop}\label{prop1}
There exists at least one real, positive steady state $\mathcal{B}=(B_\ast,P_\ast,0)$ of Eqs.~\eqref{eq:lv} and this steady state is unique if ${\delta+(\gamma-1)\kappa>0}$ or $\gamma\varsigma+(\delta-\mu)\varpi>0$.
\end{prop}
\noindent From scalings~\eqref{eq:scalings}, we expect both of these conditions to be satisfied. In what follows, we can therefore assume that both equilibria in Eqs.~\eqref{eq:eqs} exist and are unique.

\begin{figure}[b]
\centering\includegraphics{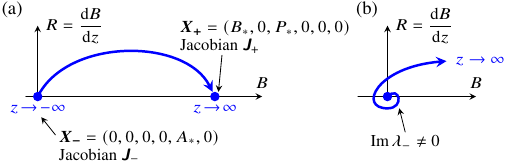}
\caption{Necessary conditions for the existence of a travelling wave. (a)~A travelling wave defines a heteroclinic connection in the space of $\vec{X}=(B,R,P,S,A,T)$: plot of a projection onto the $(B,R)$ plane, defining the steady states $\vec{X}_\pm$ and the Jacobians $\tens{J}_\pm$ at $z\to\pm\infty$. (b)~Existence of a travelling requires the stable eigenvalue of $\tens{J}_\pm$ that dominates as $z\to\pm\infty$ to be real for the solution to be nonnegative. See text for further explanation.}\label{fig2}
\end{figure}

\subsection{Exact calculations}
We seek travelling-wave solutions of Eqs.~\eqref{eq:lv}, for which purpose we recast these into six first-order ordinary differential equations in terms of the travelling-wave coordinate ${z=x-ct}$, where $c$ is the wave velocity \figref{fig1}{c}. Thus $c>0$ corresponds to invasion by the competitors, while $c<0$ is invasion of the competitors. Now
\begin{subequations}\label{eq:tweq}
\begin{align}
B'&=R,\\
R'&=-cR-B(1 - \alpha A - B - \kappa P ) + \beta AB\!+\!\gamma B\!-\!\delta P,\\
P'&=S,\\
dS'&=-cS-P(\mu - \xi A - \varpi B - \varsigma P ) - \beta AB- \gamma B + \delta P,\\
A'&=T,\\
DT'&=-cT-A(\zeta - \eta A - \vartheta B - \iota P ),
\end{align}
\end{subequations}
in which dashes denote differentiation with respect to $z$. We further introduce the solution vector $\vec{X}=(B,R,P,S,A,T)$.

As announced, we will specifically seek travelling waves in which one species invades the other, i.e., in which the equilibrium~$\mathcal{A}$ of competitors only, at $z\to-\infty$, invades ($c>0$) or is invaded by (${c<0}$) the equilibrium~$\mathcal{B}$ without competitors, at $z\to\infty$ [\textfigref{fig1}{c} and Eqs.~\eqref{eq:eqs}]. In phase space \figref{fig2}{a}, this defines a heteroclinic connection from the equilibrium $\vec{X_-}=(0,0,0,0,A_\ast,0)$ of Eqs.~\eqref{eq:tweq} to their equilibrium ${\vec{X_+}=(B_\ast,0,P_\ast,0,0,0)}$.

Determining the speed selected by travelling-wave solutions for given initial conditions in the partial differential equations~\eqref{eq:lv} is, mathematically, a problem in analysis. It is a higher-dimensional analogue of the problem of wave-speed selection in the Fisher--KPP equation~\cite{murray,vansarloos}. However, just as there is an algebraic bound $c\geq c_{\min}$ on the wave speed in the Fisher--KPP equation~\cite{murray,simpson24}, we will discover analogous algebraic constraints on travelling-wave solutions of Eqs.~\eqref{eq:lv}. We will focus on these algebraic necessary conditions, here, and complement our calculations with numerical solutions in the next subsection. The following calculations generalise results for travelling waves for two competing species without phenotypic switching that we summarise in \citeappendix{appB}.
\begin{widetext}
\subsubsection{Necessary conditions for the existence of a travelling wave}
To obtain necessary conditions for the existence of a travelling wave, we consider the dynamics for $z\to\pm\infty$. There, near their equilibria $\vec{X_\pm}$, Eqs.~\eqref{eq:tweq} linearise to $\vec{X'}=\tens{J}_\pm\cdot\vec{X}$, where the Jacobians $\tens{J}_\pm$ are given by
\begin{subequations}\label{eq:Js}
\begin{align}
\tens{J}_-&=\begin{pmatrix}
0 & 1 & 0 & 0 & 0 & 0 \\
\gamma-1+(\alpha+\beta)\dfrac{\zeta}{\eta} & -c & -\delta & 0 & 0 & 0 \\
0 & 0 & 0 & 1 & 0 & 0 \\
-\dfrac{1}{d}\left(\dfrac{\beta \zeta}{\eta} + \gamma\right) & 0 & \dfrac{1}{d}\left(\delta-\mu+\dfrac{\xi\zeta}{\eta}\right) & -\dfrac{c}{d} & 0 & 0 \\
0 & 0 & 0 & 0 & 0 & 1 \\
\dfrac{\vartheta\zeta}{D\eta} & 0 & \dfrac{\iota\zeta}{D\eta} & 0 & \dfrac{\zeta}{D} & -\dfrac{c}{D}
\end{pmatrix},\label{eq:Jminus}\\
\tens{J}_+&=\begin{pmatrix}
0 & 1 & 0 & 0 & 0 & 0 \\
\gamma-1 + 2B_\ast + \kappa P_\ast & -c & \kappa B_\ast - \delta & 0 & (\alpha+\beta)B_\ast & 0 \\
0 & 0 & 0 & 1 & 0 & 0 \\
\dfrac{\varpi P_\ast - \gamma}{d} & 0 & \dfrac{\delta-\mu + \varpi B_\ast + 2\varsigma P_\ast}{d} & -\dfrac{c}{d} & \dfrac{\xi P_\ast - \beta B_\ast}{d} & 0 \\
0 & 0 & 0 & 0 & 0 & 1 \\
0 & 0 & 0 & 0 & \dfrac{\vartheta B_\ast + \iota P_\ast-\zeta}{D} & -\dfrac{c}{D}
\end{pmatrix},\label{eq:Jplus}
\end{align}
\end{subequations}
\end{widetext}
in which $B_\ast,P_\ast$ are determined by Eqs.~\eqref{eq:BPeq}. Jacobians $\tens{J}_-$ and $\tens{J}_+$ have block structures so that the eigenvalues of each matrix are determined by one quadratic and one quartic equation each. The latter cannot be solved usefully in closed form, but we can obtain expansions for their roots for $\varepsilon\ll 1$, assuming that $c=O(1)$. In this way, we find that the eigenvalues of $\tens{J}_-$ are
\begin{subequations}\label{eq:evA}
\begin{align}
\lambda_-&=\dfrac{1}{2D}\left(-c\pm\sqrt{c^2+4D\zeta}\right),\label{eq:evA1}\\
\lambda_-&=\dfrac{1}{2}\left[-c\pm\sqrt{c^2+4\left(\dfrac{\alpha\zeta}{\eta}-1\right)}\right]+O(\varepsilon),\label{eq:evA2}\\
\lambda_-&=\dfrac{\delta}{c}+O(\varepsilon),\qquad\lambda_-=-\dfrac{c}{d}+O(1).\label{eq:evA3}
\end{align}
\end{subequations}
The eigenvalues of $\tens{J}_+$ are
\begin{subequations}\label{eq:evB}
\begin{align}
\lambda_+&=\dfrac{1}{2D}\left[-c\pm\sqrt{c^2+4D\left(\vartheta B_\ast+\iota P_\ast-\zeta\right)}\right],\label{eq:evB1}\\
\lambda_+&=\dfrac{1}{2}\left(-c\pm\sqrt{c^2+4}\right)+O(\varepsilon),\label{eq:evB2}\\
\lambda_+&=\dfrac{\delta}{c}+O(\varepsilon),\qquad\lambda_+=-\dfrac{c}{d}+O(1),\label{eq:evB3}
\end{align}
\end{subequations}
where we have used $B_\ast=1+O\bigl(\varepsilon^2\bigr)$, $P_\ast=O(\varepsilon)$ to obtain Eq.~\eqref{eq:evB2}. We note that the higher-order corrections to these eigenvalues are determined by solving linear equations at each order. This means that if the leading-order term in the expansion of an eigenvalue is real, then so is the eigenvalue itself.

Now Eqs.~\eqref{eq:tweq} define a sixth-order system of ordinary differential equations, so we may impose 6 boundary conditions. However, the travelling-wave problem is invariant under translations in $z$, so one of these must remove this freedom of translation. A necessary condition for the existence of the heteroclinic connection in \textfigref{fig2}{a} is therefore that the remaining 5 boundary conditions eliminate the growing modes as $z\to\pm\infty$. These growing modes are associated with eigenvalues $\lambda_\pm$ in the eigenspaces $\mathcal{E}(\tens{J}_\pm)$ such that $\operatorname{Re}{\lambda_\pm}\gtrless0$. An algebraically necessary condition for the existence of a travelling wave is thus
\begin{align}
N_++N_-=5,\quad\text{where }N_\pm=\left|\{\lambda_\pm\in\mathcal{E}(\tens{J}_\pm)\mid\operatorname{Re}{\lambda_\pm}\gtrless 0\}\right|.
\end{align}
From Eq.~\eqref{eq:evA}, and assuming that $\alpha\zeta/\eta-1=O(1)$, we obtain
\begin{subequations}
\begin{align}
N_-&=\left\{\begin{array}{cl}
2&\text{if $\zeta<\dfrac{\eta}{\alpha}$ and $c<0$}, \\[3mm]
3&\text{if $\zeta>\dfrac{\eta}{\alpha}$},\\[3mm]
4&\text{if $\zeta<\dfrac{\eta}{\alpha}$ and $c>0$}.
\end{array}\right.\label{eq:Nminus}
\end{align}
In particular, although Eq.~\eqref{eq:evA2} has corrections at order $O(\varepsilon)$, these do not affect the sign of $\operatorname{Re}{\lambda_-}$ if $\alpha\zeta/\eta-1=O(1)$. Hence the conditions in Eq.~\eqref{eq:Nminus} are exact.

Next, we compute $\vartheta B_\ast+\iota P_\ast-\zeta=\vartheta-\zeta+O\bigl(\varepsilon^2\bigr)$ using $B_\ast=1+\smash{O\bigl(\varepsilon^2\bigr)}$, $P_\ast=\gamma/\delta+\smash{O\bigl(\varepsilon^2\bigr)}$, so, from Eqs.~\eqref{eq:evB},
\begin{align}
N_+&=\left\{\begin{array}{cl}
2&\text{if $\zeta>\vartheta+O\bigl(\varepsilon^2\bigr)$ and $c>0$}, \\[3mm]
3&\text{if $\zeta<\vartheta+O\bigl(\varepsilon^2\bigr)$},\\[3mm]
4&\text{if $\zeta>\vartheta+O\bigl(\varepsilon^2\bigr)$ and $c<0$}.
\end{array}\right.\label{eq:nplus}
\end{align}
This implies that
\begin{align}
N_++N_-=5\quad \text{if}\left\{\begin{array}{l}
\zeta<\min{\left\{\vartheta+O\bigl(\varepsilon^2\bigr),\dfrac{\eta}{\alpha}\right\}}\text{ and }c<0\text{ or} \\[3mm]
\zeta>\max{\left\{\vartheta+O\bigl(\varepsilon^2\bigr),\dfrac{\eta}{\alpha}\right\}}\text{ and }c>0,
\end{array}\right.\label{eq:nc}
\end{align}
\end{subequations}
and $N_++N_->5$ otherwise. This shows that a travelling wave of invasion is possible in two cases: invasion by competitors ($c>0$) is possible only if the growth rate of competitors is sufficiently large, $\zeta>\max{\left\{\vartheta+O\bigl(\varepsilon^2\bigr),\eta/\alpha\right\}}$; invasion of the competitors ($c<0$) is possible only if their growth rate is sufficiently small, $\zeta<\min{\left\{\vartheta+O\bigl(\varepsilon^2\bigr),\eta/\alpha\right\}}$. 

We have actually proved slightly a slightly more general result: in Eqs.~\eqref{eq:nplus} and~\eqref{eq:nc}, $\vartheta+O\bigl(\varepsilon^2\bigr)$ can be replaced with its exact value $\vartheta B_\ast+\iota P_\ast$.

\subsubsection{Minimum wavespeed of travelling waves}
We can now obtain a lower bound, $|c|\geqslant |c_\text{min}|$, on the wavespeed of these travelling waves. The key observation is the following: In the $(B,R)$ or $(P,S)$ planes, the dynamics near $\vec{X_-}$ are dominated by the stable eigenmode of $\tens{J}_-$ that has the slowest decay as $z\to-\infty$, since the $N_-$ boundary conditions have removed the unstable eigenmodes of $\tens{J}_-$. If the eigenvalue $\lambda_-$ associated to this mode were complex, then the dynamics would describe a spiral in the $(B,R)$ or $(P,S)$ planes~\figref{fig2}{b}, which would contradict $B(z),P(z)\geq 0$. Hence $\lambda_-$ must be real. Similarly, the stable eigenmode of $\tens{J}_+$ that dominates the dynamics in the $(A,T)$ plane as $z\to\infty$ must be real lest $A(z)\geq 0$ break.

In the limit $\varepsilon\ll1$, only eigenvalues~\eqref{eq:evA2} and~\eqref{eq:evB1} can be complex. Indeed, eigenvalues~\eqref{eq:evA2}, \eqref{eq:evA3}, \eqref{eq:evB2}, and~\eqref{eq:evB3} are real to leading order and hence to all orders, since, as noted above, their corrections are determined by linear equations that cannot introduce any additional imaginary terms. 

Now, if $c>0$ (invasion by the competitors), then ${\zeta>\eta/\alpha}$ from Eq.~\eqref{eq:nc}, so eigenvalues~\eqref{eq:evA2} are real. Moreover, eigenvalues~\eqref{eq:evB1} are stable ($\operatorname{Re}{\lambda_+}<0$) because ${\zeta>\vartheta B_\ast+\iota P_\ast}$, and they are real if and only if $c^2\geqslant 4D(\zeta-\vartheta B_\ast-\iota P_\ast)$. They dominate the dynamics in the $(A,T)$ plane: the eigenvectors corresponding to eigenvalues~\eqref{eq:evB2} and~\eqref{eq:evB3} are orthogonal to this plane because of the block matrix structure of $\tens{J}_+$ exhibited by Eq.~\eqref{eq:Jplus}. This proves
\begin{align}
c_\text{min}&=2\sqrt{D(\zeta-\vartheta B_\ast-\iota P_\ast)}\nonumber\\
&=2\sqrt{D(\zeta-\vartheta)}+O\bigl(\varepsilon^2\bigr)\quad \text{if }c>0.\label{eq:cminplus}
\end{align}
Remarkably, this result is independent of the rate of responsive switching, $\beta$, at all orders, because Eqs.~\eqref{eq:BPeq}, which determine $B_\ast,P_\ast$, are independent of $\beta$, too. This result holds not only in the limit $\varepsilon\ll 1$ assumed above. In Appendix~\ref{appA}, we prove, more generally, the following:
\begin{prop}\label{prop2}
If ${\delta+(\gamma-1)\kappa>0}$ or $\gamma\varsigma+(\delta-\mu)\varpi>0$, then $c_\text{min}=2\smash{\sqrt{D(\zeta-\vartheta B_\ast-\iota P_\ast)}}$ in any travelling wave of invasion by the competitors ($c>0$).
\end{prop}
\noindent Similarly, in the case $c<0$ (invasion of the competitors), ${\zeta<\vartheta+O\bigl(\varepsilon^2\bigr)}$ and $\zeta<\eta/\alpha$ from Eq.~\eqref{eq:nc}. Hence eigenvalues Eq.~\eqref{eq:evB1} are real. However, eigenvalues~\eqref{eq:evA2} are stable [$\operatorname{Re}{\lambda_-}=O(1)>0$], and they dominate the dynamics in the $(B,R,P,S)$ plane, because, using the block matrix structure in Eq.~\eqref{eq:Jminus}, the eigenvectors of $\tens{J}_-$ corresponding Eqs.~\eqref{eq:evA1} are orthogonal to this plane; the one other stable eigenvalue is, from Eq.~\eqref{eq:evA3}, $\lambda_-\sim |c|/d=O\bigl(\varepsilon^{-1}\bigr)$, so its eigenmode decays faster. Hence eigenvalues~\eqref{eq:evA2} must be real, which requires $c^2\geqslant 4(1-\alpha\zeta/\eta)+O(\varepsilon)$. Thus
\begin{subequations}
\begin{align}
c_\text{min}=-2\sqrt{1-\dfrac{\alpha\zeta}{\eta}}+O(\varepsilon)\quad\text{if }c<0.
\end{align}
We will need the first correction to this leading-order result in what follows. To compute this correction, we note that, at the critical value $c=c_\text{min}$, the eigenvalues in Eqs.~\eqref{eq:evA2} merge. Hence the quartic polynomial determining the eigenvalues in Eqs.~\eqref{eq:evA2} and~\eqref{eq:evA3} has a double root, so its discriminant~\cite{algebra} vanishes. Using \textsc{Mathematica} (Wolfram, Inc.), we compute this discriminant to be an eighth-order polynomial in $c$. Seeking a solution $\smash{c=-2\sqrt{1-\alpha\zeta/\eta}+\varepsilon c_1+O\bigl(\varepsilon^2\bigr)}$, we obtain
\begin{align}
c_\text{min}=-2\sqrt{1-\dfrac{\alpha\zeta}{\eta}}\left[1-\dfrac{\beta\zeta+\gamma\eta}{\delta\eta+2(\eta-\alpha\zeta)}+O\bigl(\varepsilon^2\bigr)\right]\;\text{ if }c<0.\label{eq:cminminus}
\end{align}
\end{subequations}
\subsubsection{Averaged model: effect of phenotypic variation}
To analyse the effect of phenotypic variation on these travelling waves of invasion, we need to compare our results to results for an averaged model~\cite{haas20,haas22} without phenotypic variation, with two effective populations $\bar{B}(x,t)$, corresponding to bacteria or persisters, and $\bar{A}(x,t)$, corresponding to competitors. Their effective dynamics are
\begin{subequations}\label{eq:lvav}
\begin{align}
\dfrac{\partial\bar{B}}{\partial t}=\bar{B}\bigl(\bar{\omega}-\bar{\alpha}\bar{A}-\bar{\chi}\bar{B}\bigr)+\bar{D}_B\dfrac{\partial^2\bar{B}}{\partial x^2},\\
\dfrac{\partial\bar{A}}{\partial t}=\bar{A}\bigl(\bar{\zeta}-\bar{\eta}\bar{A}-\bar{\vartheta}\bar{B}\bigr)+\bar{D}_A\dfrac{\partial^2\bar{A}}{\partial x^2},
\end{align}
\end{subequations}
which have single-species equilibria $\bar{\mathcal{A}}=\bigl(0,\bar{A}_\ast\bigr)$, $\bar{\mathcal{B}}=\bigl(\bar{B}_\ast,0\bigr)$, with $\bar{A}_\ast=\bar{\zeta}/\bar{\eta}$, $\bar{B}_\ast=\bar{\omega}/\bar{\chi}$. The effective parameters $\bar{\omega},\bar{\alpha},\bar{\chi}$, $\bar{\zeta},\bar{\eta},\bar{\vartheta},\bar{D}_A,\bar{D}_B$ are determined by consistency conditions~\cite{haas22} imposing that the populations, births, and competitions at these equilibria match the populations, births, and competitions at the corresponding equilibria of the full model~\eqref{eq:lv}. Thus
\begin{subequations}\label{eq:ccs}
\begin{align}
\bar{B}_\ast&=B_\ast+P_\ast,&\bar{A}_\ast&=A_\ast,\label{eq:ccpop}\\
\bar{\omega}\bar{B}_\ast&=B_\ast+\mu P_\ast,&\bar{\zeta}\bar{A}_\ast&=\zeta A_\ast,
\end{align}
and
\begin{align}
\bar{\alpha}\bar{B}_\ast\bar{A}_\ast&=\alpha B_\ast A_\ast+\xi P_\ast A_\ast,\\
\bar{\vartheta}\bar{A}_\ast\bar{B}_\ast&=\vartheta A_\ast B_\ast+\iota A_\ast P_\ast,\\
\bar{\chi}\bar{B}_\ast^2&=B_\ast^2+(\kappa+\varpi)B_\ast P_\ast+\varsigma P_\ast^2,\\
\bar{\eta}\bar{A}_\ast^2&=\eta A_\ast^2.
\end{align}
Additionally, averaging the diffusivities yields
\begin{align}
\bar{D}_B\bar{B}_\ast&=B_\ast+d P_\ast,&&\bar{D}_A\bar{A}_\ast=DA_\ast.\label{eq:Dav}
\end{align}
\end{subequations}
A more refined model would note that, if bacteria move at (average) speed $v$ with persistence time $\tau$, their diffusivity scales, by dimensional analysis, like $v^2/\tau$. Averaging velocities instead would modify Eqs.~\eqref{eq:Dav}, but would still yield $\bar{D}_A=D$, though it would change $\bar{D}_B$ at order $O(\varepsilon)$. 

Given the full model~\eqref{eq:lv}, Eqs.~\eqref{eq:ccs} determine the parameters of the averaged model~\eqref{eq:lvav} uniquely. Now, as shown in \citeappendix{appB}, the velocity $c$ of travelling wave solutions of Eqs.~\eqref{eq:lvav} satisfies $|c|\geqslant |\bar{c}_\text{min}|$, where
\begin{align}
\bar{c}_\text{min}=\left\{\begin{array}{cl}
-2\sqrt{\bar{D}_B\left(\bar\omega-\dfrac{\bar\alpha\bar\zeta}{\bar\eta}\right)}&\text{if $c<0$},\\
2\sqrt{\bar D_A\left(\bar\zeta-\dfrac{\bar\vartheta\bar\omega}{\bar\chi}\right)}&\text{if $c>0$}.
\end{array}\right.\label{eq:cminav}
\end{align}
We begin by considering the case $c>0$ (invasion by the competitors). Adding Eqs.~\eqref{eq:BPeq} yields
\begin{subequations}
\begin{align}
B_\ast+\mu P_\ast=B_\ast^2+(\kappa+\varpi)B_\ast P_\ast+\varsigma P_\ast^2.
\end{align}
Using Eqs.~\eqref{eq:ccs}, we now compute
\begin{align}
\dfrac{\bar\vartheta\bar\omega}{\bar\chi}&=\dfrac{\dfrac{\vartheta B_\ast+\iota P_\ast}{B_\ast+P_\ast}\dfrac{B_\ast+\mu P_\ast}{B_\ast+P_\ast}}{\dfrac{B_\ast^2+(\kappa+\varpi)B_\ast P_\ast+\varsigma P_\ast^2}{(B_\ast+P_\ast)^2}}=\vartheta B_\ast+\iota P_\ast,
\end{align}
\end{subequations}
whence
\begin{align}
\bar{c}_\text{min}=2\sqrt{D\left(\zeta-\vartheta B_\ast-\iota P_\ast\right)}=c_\text{min}\quad\text{if }c>0,
\end{align}
from Eq.~\eqref{eq:cminplus}, i.e., the minimum wavespeed in the averaged model matches that of the full model if $c>0$. [Since the result is independent of $\bar D_B$, this conclusion is unaffected by modifications of Eqs.~\eqref{eq:Dav}.] In other words, we have shown that phenotypic switching has no effect on the speed of invasion by the competitors.

Finally, we turn to the case $c<0$ (invasion of the competitors). We now compute, using Eqs.~\eqref{eq:ccs},
\begin{align}
\bar{c}_\text{min}&=-\dfrac{2}{B_\ast+P_\ast}\sqrt{(B_\ast+d P_\ast)\left[\left(1\!-\!\dfrac{\alpha\zeta}{\eta}\right)B_\ast+\left(\mu\!-\!\dfrac{\xi\zeta}{\eta}\right)P_\ast\right]}\nonumber\\
&=-2\sqrt{1-\dfrac{\alpha\zeta}{\eta}}\left[1-\dfrac{\gamma}{\delta}{}+O\bigl(\varepsilon^2\bigr)\right]\neq c_\text{min}\quad\text{if }c<0,\label{eq:cbarminminus}
\end{align}
by comparison with Eq.~\eqref{eq:cminminus}. In other words, phenotypic switching does effect the speed of invasion of the competitors. [Again, this conclusion is not affected by modifications of Eqs.~\eqref{eq:Dav}, because $c_\text{min}$, given by Eq.~\eqref{eq:cminminus}, depends on $\beta$, while $\bar{c}_\text{min}$ cannot depend on $\beta$ because neither Eqs.~\eqref{eq:BPeq} nor Eqs.~\eqref{eq:ccs} depend on $\beta$.]

From Eqs.~\eqref{eq:cminminus} and~\eqref{eq:cbarminminus} and recalling that $\eta>\alpha\zeta$ if $c<0$, we find that $|c_\text{min}|>|\bar{c}_\text{min}|$ if $\beta<2\gamma(\eta-\alpha\zeta)/(\delta\zeta)$. In other words, phenotypic switching can speed up invasion of the competitors, but only if the rate of responsive switching is not too high. We conclude by noting that this bound might change if Eqs.~\eqref{eq:Dav} are modified.

\begin{figure}[t!]
\centering\includegraphics{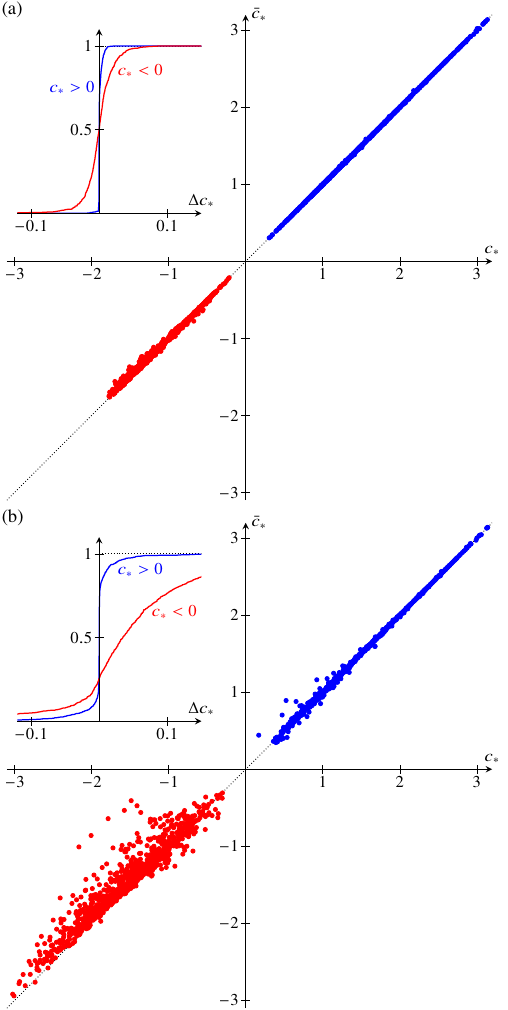}
\caption{Numerical results. (a) Plot of the velocity $\bar{c}_\ast$ of travelling waves of invasion in the averaged model~\eqref{eq:lvav} against their velocity~$c_\ast$ in the full model~\eqref{eq:lv}. Each mark corresponds to one of ${\mathcal{N}=1000}$ randomly sampled systems obeying the persister scalings~\eqref{eq:scalings} in each of the cases $c_\ast>0$ (blue, invasion by the competitors) and $c_\ast<0$ (red, invasion of the competitors). Inset: corresponding cumulative distribution functions of the relative difference $\upDelta c_\ast$ of the wave velocities. (b)~Corresponding results for general, randomly sampled systems. See text for further explanation. Parameters values for the numerical\linebreak calculations (\citeappendix{appC}): $f\!=\!2$, $k\!=\!1$, $M\!=\!100$, $T\!=\!5$, (a)~$\varepsilon\!=\!0.1$.}\label{fig3}
\end{figure}

\subsection{Numerical results}
From general initial conditions, we expect the dynamics of the full model~\eqref{eq:lv} and the averaged model~\eqref{eq:lvav} to converge to waves of invasion travelling at velocities $c_\ast$ and $\bar{c}_\ast$ respectively. In the context of the Fisher--KPP equation, shallowness of initial conditions or nonlinearities leading to pushed waves \cite{murray,vansarloos,simpson24} can cause the solution of the Fisher--KPP equation to select a wave speed higher than that of the pulled wave determined by the linearised dynamics as $z\to\pm\infty$. In the same way,  $c_\ast$ and $\bar{c}_\ast$ may be different to the linear velocities $c_\text{min}$ and $\bar{c}_\text{min}$ analysed in the previous subsection. 

Even in the two-species \mbox{Lotka--Volterra system~\cite{lewis02,fei03,alhasanat19}}, necessary and sufficient conditions for the solution to converge to these linear wavespeeds are not known. For this reason, we compute the wave velocities $c_\ast$ and $\bar{c}_\ast$ numerically, as described in \citeappendix{appC}, and we randomly sample the two species, i.e., the model parameters $\alpha, \beta,\gamma,\delta,\zeta,\eta,\vartheta$, $\iota,\kappa,\mu,\xi,\varpi,\varsigma$ and $d,D$ (\citeappendix{appC}).

We begin by sampling parameter values obeying the persister scalings~\eqref{eq:scalings}, and plot $\bar{c}_\ast$ against $c_\ast$ in \textfigref{fig3}{a}. The data cluster near the diagonal $\bar{c}_\ast=c_\ast$ of the plot, especially strongly so for $c_\ast>0$. We quantify the deviation from this diagonal by the relative difference $\upDelta c_\ast=1-\bar{c}_\ast/c_\ast$. Its cumulative distribution functions for $c_\ast\gtrless0$ \figrefi{fig3}{a} show that the distribution of $\upDelta c_\ast$ is less localised for invasion of the competitors ($c_\ast<0$) than for invasion by the competitors ($c_\ast>0$). Moreover, the distributions show that $\Delta c_\ast\gtrless 0$ and hence $|c_\ast|\gtrless|\bar c_\ast|$ are almost equally likely.

We extend these results by sampling general parameter values, not necessarily obeying scalings~\eqref{eq:scalings} or indeed the conditions of Propositions~\ref{prop1} and~\ref{prop2}. The data~\figref{fig3}{b} continue to cluster near $\bar{c}_\ast=c_\ast$ for $c_\ast>0$, but the deviation from this diagonal increases for $c_\ast<0$, in which case the distribution has more mass in $\Delta c_\ast>0$ \figrefi{fig3}{b}.

In this way, our numerical results confirm the conclusions that we have drawn from the exact calculations in the previous subsection: phenotypic switching has little if any effect on the speed of invasion by the competitors ($c_\ast>0$), but it does affect invasion of the competitors ($c_\ast<0$). In the latter case, phenotypic switching can even speed up invasion of the competitors on average ($|c_\ast|>|\bar{c}_\ast|$).

Nevertheless, it is important to note that the wave velocities to which the initial conditions converge differ in general from the minimum wave velocities: our numerical simulations suggest that $|c_\ast|>|c_\text{min}|$ \wholefigref{fig4} and $|\bar{c}_\ast|>|\bar{c}_\text{min}|$ (not shown), especially for slow waves. We have argued above that this is not surprising, but analysing the underlying mathematical mechanisms, such as initial conditions or pushed waves, is beyond the scope of this work.

\begin{figure}[t]
\centering\includegraphics{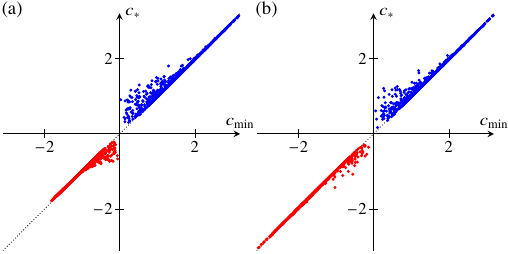}
\caption{Numerical results. Plot of the velocity $c_\ast$ of travelling waves in the full model~\eqref{eq:lv} against the minimum wave velocity $c_\text{min}$ determined by the linearised dynamics for (a) randomly sampled systems obeying the persister scalings~\eqref{eq:scalings} and (b) general randomly sampled systems, corresponding to those analysed in \textwholefigref{fig3}.}\label{fig4}
\end{figure}

\section{\uppercase{Conclusion}}
In this paper, we have studied the effect of stochastic and responsive phenotypic switching on the travelling waves by which one species invades another in a minimal two-species setting opposing one species switching between two phenotypes to a competitor species. Through exact calculations of the minimum wave speed determined by the linearised dynamics near the one-species steady states and numerical simulations of the travelling waves which the dynamics converge to, we have found that phenotypic switching does not affect the speed at which the switching species is invaded by the competitors, but does speed up the wave by which the switching species invades the competitors on average.

This has a somewhat counterintuitive ecological implication: switching to a phenotype resilient to competition is not a defensive ecological strategy, but can be an offensive ecological strategy. Still, this result is consistent with our earlier theoretical work~\cite{haas22}, in which we found that switching to a rare ``attack'' phenotype can stabilise well-mixed coexistence of two species. Moreover, recent experimental~\cite{oliveira22} work provides further support for this view by reporting a ``suicidal'' bacterial phenotype: a subpopulation of \emph{Pseudomonas aeruginosa} migrates up antibiotic gradients to release bacteriocins that kill competitors. Future work, both experimental and theoretical, will therefore have to clarify the ecological role of phenotypic switching.

We have focused on one-dimensional spatial diffusion in our work. This simplification is of course convenient mathematically, but, in the ecological context of soil, in which bacteria move through the narrow channels of the porous soil medium, not unrealistic either. Still, future work should analyse these travelling of invasion in higher spatial dimensions, too. 

Finally, a more mathematical question remains open, too: what is the mechanism of wave speed selection in these models with phenotypic switching? This is, we believe, an interesting avenue for future research because phenotypic switching provides a source of immigration for each species, with abundance $X$, say: in classical population dynamics models, $\dot{X}=0$ if $X=0$, but $\dot{X}\neq0$ is possible if $X=0$ in these models, as exemplified by our minimal two-species model~\eqref{eq:lv}.

\begin{acknowledgments}
We thank Nuno M. Oliveira for discussions and Federico Stefanelli and Kunjeti Dharanidhar Gupta for preliminary computations at an early stage of this work. We are grateful for funding from the Max Planck Society.
\end{acknowledgments}

\subsection*{Data availability}
\textsc{Matlab} (The MathWorks, Inc.)~code for the numerical calculations is available openly at Ref.~\footnote{Code is available at \href{https://doi.org/10.5281/zenodo.20072527}{\texttt{doi.org/10.5281/zenodo.20072527}}.}.

\appendix
\titleformat{\section}{\normalfont\bfseries\centering\small}{APPENDIX \thesection:}{0.5em}{}
\section{\uppercase{Proofs of propositions}}\label{appA}

\subsection{Proof of Proposition~\ref{prop1}}
Equations~\eqref{eq:BPeq} define $(B_\ast,P_\ast)$ as the intersection, in the positive quadrant of the $(B,P)$ plane, of the hyperbolae
\begin{subequations}
\begin{align}
\mathcal{H}_1&\colon B^2+\kappa BP+(\gamma-1)B-\delta P=0,\\
\mathcal{H}_2&\colon \varpi BP+\varsigma P^2-\gamma B+(\delta-\mu)P=0.
\end{align}    
\end{subequations}
We notice that $\mathcal{H}_1$ has a vertical asymptote $B=\delta/\kappa$. Hence, as $B\to\delta/\kappa^-$ or $B\to\delta/\kappa^+$,  $P\to\infty$ in the first quadrant. As $B\to\infty$, $P\sim-\kappa B\to-\infty$, so this branch must enter or leave the first quadrant by crossing $B=0$ or $P=0$. Now it is easy to see that $\mathcal{H}_1$ intersects $B=0$ and $P=0$ at $(0,0)$ and $(1-\gamma,0)$ only. Moreover, at these points, we compute $\partial P/\partial B=(\gamma-1)/\delta$ and $\partial P/\partial B=(1-\gamma)/[\delta-\kappa(1-\gamma)]$, respectively. There are thus three cases:

\begin{figure}[t]
\centering\includegraphics{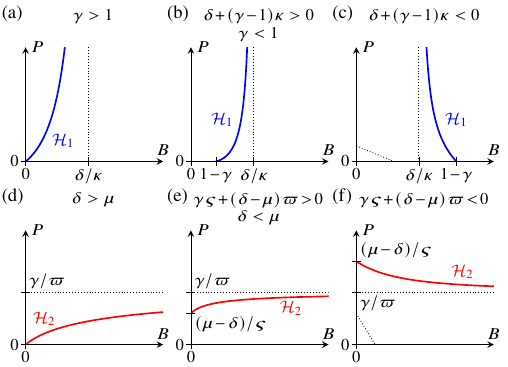}
\caption{Proof of Proposition~\ref{prop1}. Plots, in the positive quadrant of the $(B,P)$ plane, of hyperbolae $\mathcal{H}_1$ in the cases (a) $\delta+(\gamma-1)\kappa>0$ and $\gamma>1$, (b) $\delta+(\gamma-1)\kappa>0$ and $\gamma<1$, (c) $\delta+(\gamma-1)\kappa<0$ and $\mathcal{H}_2$ in the cases (d)~$\gamma\varsigma+(\delta-\mu)\varpi>0$ and $\delta>\mu$, (e) $\gamma\varsigma+(\delta-\mu)\varpi>0$ and $\delta<\mu$, (f) $\gamma\varsigma+(\delta-\mu)\varpi<0$. Dotted lines are the asymptotes to $\mathcal{H}_1$ or $\mathcal{H}_2$.}\label{fig5}
\end{figure}

\begin{enumerate}[leftmargin=*,label=(\arabic{enumi}),widest=3,topsep=4pt,itemsep=0pt]
\item If $\gamma>1$ \figref{fig5}{a}, then $\mathcal{H}_1$ must enter the first quadrant at $(0,0)$, where, indeed, $\partial P/\partial B>0$, and then approach $B=\delta/\kappa$ from below. Conversely, if $\gamma<1$, then $\mathcal{H}_1$ does not enter the first quadrant near $(0,0)$.
\item If $\gamma<1$ but $\delta+(\gamma-1)\kappa>0$ \figref{fig5}{b}, then $1-\gamma<\delta/\kappa$, and $\mathcal{H}_1$ enters the first quadrant at $(1-\gamma,0)$. Indeed, there, $\partial P/\partial B>0$, so $\mathcal{H}_1$ approaches $B=\delta/\kappa$ from below.
\item Finally, if $\delta+(\gamma-1)\kappa<0$ \figref{fig5}{c}, then $1-\gamma>\delta/\kappa$, so $\mathcal{H}_1$ approaches $B=\delta/\kappa$ from above, and leaves the first quadrant at $(1-\gamma,0)$ where $\partial P/\partial B<0$.
\end{enumerate}
An analogous argument determines the possible orientations of $\mathcal{H}_2$ [Figs.~\figrefp{fig5}{d}--\figrefp{fig5}{f}] in the positive quadrant relative to its horizontal asymptote $P=\gamma/\varpi$. These orientations depend on whether ${\gamma\varsigma+(\delta-\mu)\varpi\gtrless0}$ and $\delta\gtrless\mu$. In particular, we compute $\partial P/\partial B=\gamma/(\delta-\mu)$ at $(0,0)$.

The arcs of the hyperbolae in \textwholefigref{fig3} define smooth functions $\mathcal{P}_1(B),\mathcal{P}_2(B)$, on $[0,\delta/\kappa)$, $[1-\gamma,\delta/\kappa)$, or $(\delta/\kappa,1-\gamma]$ in the cases corresponding to Figs.~\figrefp{fig3}{a}--\figrefp{fig3}{c}, so their difference $\mathcal{P}(B)=\mathcal{P}_1(B)-\mathcal{P}_2(B)$ is smooth on this interval, too. Existence and uniqueness of $(B_\ast,P_\ast)$ is now existence and uniqueness of a root $\mathcal{P}(B_\ast)=0$. To establish this, we need the following general fact:

\begin{lemma}\label{lemma1}
Let $f\colon[a,b]\to\mathbb{R}$ be a smooth function such that $f(a)f(b)<0$.
\begin{enumerate}[label=(\alph{enumi}),widest=2,leftmargin=*,itemsep=0pt,topsep=4pt]
\item There exists $\xi\in(a,b)$ such that $f(\xi)=0$.
\item This root is unique if (i) $f'(x)\neq 0$ for all $x\in (a,b)$ or (ii)~$f(a)\gtrless 0$ and $f''(x)\lessgtr 0$ for all $x\in (a,b)$.
\end{enumerate}
\end{lemma}
\vspace{-9pt}
\begin{proof}
Part~(a) is a direct consequence of the intermediate value theorem~\cite{garling}. For part (b), suppose that there exist $\xi_1,\xi_2\in [a,b]$ such that $f(\xi_1)=f(\xi_2)=0$. Hence, by Rolle's theorem~\cite{garling}, there exists $v\in(\xi_1,\xi_2)\subset(a,b)$ such that $f'(v)=0$, which contradicts (i). Now, if $f(a)\gtrless0$, then, using the mean-value theorem~\cite{garling} twice, there exists $u\in(a,\xi_1)$ such that $f'(u)\lessgtr 0$ and hence $w\in(u,v)\subset(a,b)$ such that $f''(w)\gtrless 0$, which contradicts (ii).
\end{proof}
\vspace{-5pt}
\noindent We now apply Lemma~\ref{lemma1} in four cases defined by \textwholefigref{fig5} to establish existence and, in three of these cases, uniqueness of the root $\mathcal{P}(B_\ast)=0$:
\begin{enumerate}[leftmargin=*,label=(\arabic{enumi}),widest=3,topsep=4pt,itemsep=0pt]
\item If $\gamma>1$, $\delta>\mu$ [Figs.~\figrefp{fig5}{a} and~\figrefp{fig5}{d}], then we take ${\epsilon\ll 1}$ and consider the restriction $\mathcal{P}\colon(\epsilon,\delta/\kappa-\epsilon)\to\mathbb{R}$. From our previous results,
\begin{align}
\mathcal{P}(\epsilon)&=\left(\dfrac{\gamma-1}{\delta}-\dfrac{\gamma}{\delta-\mu}\right)\epsilon+o(\epsilon)\nonumber\\
&=-\dfrac{(\delta-\mu)+\gamma\mu}{\delta(\delta-\mu)}\epsilon+o(\epsilon)<0.
\end{align}
In the other cases with $\delta+(\gamma-1)\kappa>0$, $\gamma\varsigma+(\delta-\mu)\varpi>0$ [Figs.~\figrefp{fig5}{a} or~\figrefp{fig5}{b} and~\figrefp{fig5}{d} or~\figrefp{fig5}{e}], we similarly consider $\mathcal{P}\colon(\max{\{0,1-\gamma\}},\delta/\kappa-\epsilon)\to\mathbb{R}$. In these cases, clearly, $\mathcal{P}(\max{\{0,1-\gamma\}})<0$. In all cases [Figs.~\figrefp{fig5}{a} or~\figrefp{fig5}{b} and~\figrefp{fig5}{d} or~\figrefp{fig5}{e}], $\mathcal{P}_1(\delta/\kappa-\epsilon)>\gamma/\varpi>\mathcal{P}_2(\delta/\kappa-\epsilon)$. Also, $\mathcal{P}''_1(B)>0$, $\mathcal{P}_2''(B)<0$ , so $\mathcal{P}''(B)>0$. Hence $\mathcal{P}$ has a unique root by Lemma~\ref{lemma1}.
\item If $\delta+(\gamma-1)\kappa>0$, $\gamma\varsigma+(\delta-\mu)\varpi<0$ [Figs.~\figrefp{fig5}{a} or~\figrefp{fig5}{b} and~\figrefp{fig5}{f}], $\mathcal{P}(\max{\{0,1-\gamma\}})<0$ and ${\mathcal{P}(\delta/\kappa-\epsilon)>0}$, again. Now $\mathcal{P}_1'(B)>0$ and $\mathcal{P}_2'(B)<0$, so ${\mathcal{P}'(B)\neq 0}$, whence $\mathcal{P}$ has a unique root by Lemma~\ref{lemma1}.
\item If $\delta+(\gamma-1)\kappa<0$, $\gamma\varsigma+(\delta-\mu)\varpi>0$ [Figs.~\figrefp{fig5}{c} and~\figrefp{fig5}{d} or~\figrefp{fig5}{e}], we consider ${\mathcal{P}\colon(\delta/\kappa+\epsilon,1-\gamma)\to\mathbb{R}}$. It is clear that ${\mathcal{P}_1(\delta/\kappa+\epsilon)>\gamma/\varpi>\mathcal{P}_2(\delta/\kappa+\epsilon)}$ and ${\mathcal{P}_1(1-\gamma)=0<\mathcal{P}_2(1-\gamma)}$. Moreover, ${\mathcal{P}'(B)\neq 0}$ because ${\mathcal{P}_1'(B)<0}$ and $\mathcal{P}_2'(B)>0$, so $\mathcal{P}$ again has a unique root by Lemma~\ref{lemma1}.
\item Finally, if $\delta+(\gamma-1)\kappa<0$, $\gamma\varsigma+(\delta-\mu)\varpi<0$ [Figs.~\figrefp{fig5}{c} and~\figrefp{fig5}{f}], then $\mathcal{P}_1(\delta/\kappa+\epsilon)>\mathcal{P}_2(\delta/\kappa+\epsilon)$ for $\epsilon\ll 1$, but ${\mathcal{P}_1(1-\gamma)=0<\mathcal{P}_2(1-\gamma)}$. This proves that $\mathcal{P}$ has a root by Lemma~\ref{lemma1}, but numerical examples (not shown) demonstrate that this need not be unique.
\end{enumerate}
This shows that there exists a root $\mathcal{P}(B_\ast)=0$, and that this is unique except in case~(4), i.e., is unique if $\delta+(\gamma-1)\kappa>0$ or $\gamma\varsigma+(\delta-\mu)\varpi>0$. This proves Proposition~\ref{prop1}.\hfill$\Box$

\subsection{Proof of Proposition~\ref{prop2}}
To prove Proposition~\ref{prop2}, we analyse the eigenvalues of the Jacobians $\tens{J}_\pm$, given by Eqs.~\eqref{eq:Js}. Both Jacobians have a block matrix structure,
\begin{align}
\tens{J}_\pm=\left(\begin{array}{c|c}
\tens{A}_\pm&\tens{B}_\pm\\
\hline
\tens{C}_\pm&\tens{D}_\pm
\end{array}\right),
\end{align}
where $\tens{B}_-=\tens{C}_+=\tzero$. The eigenvalues of $\tens{J}_\pm$ are therefore the eigenvalues of $\tens{A}_\pm$ and $\tens{D}_\pm$, respectively. The eigenvalues of~$\tens{D}_\pm$ are determined by quadratic equations, and are given by Eqs.~\eqref{eq:evA1} and~\eqref{eq:evA2}. We will therefore focus particularly on the eigenvalues of
\begin{align}
\tens{A}_\pm=\begin{pmatrix}
0&1&0&0\\
a_{11}^\pm&-c&a_{12}^\pm&0\\
0&0&0&1\\
\dfrac{a_{21}^\pm}{d}&0&\dfrac{a_{22}^\pm}{d}&-\dfrac{c}{d}
\end{pmatrix}.\label{eq:Apm}
\end{align}
From Eqs.~\eqref{eq:BPeq},
\begin{subequations}
\begin{align}
\gamma-1+2B_\ast+\kappa P_\ast&=\dfrac{\delta P_\ast}{B_\ast}+B_\ast,\\
\delta-\mu+\varpi B_\ast+2\varsigma P_\ast&=\dfrac{\gamma B_\ast}{P_\ast}+\varsigma P_\ast,
\end{align}
\end{subequations}
so, in Eq.~\eqref{eq:Apm}, by comparison with Eqs.~\eqref{eq:Js},
\begin{subequations}\label{eq:acoeffs}
\begin{align}
a_{11}^-&=\gamma-1+(\alpha+\beta)\dfrac{\zeta}{\eta},&a_{11}^+&=\dfrac{\delta P_\ast}{B_\ast}+B_\ast,\\
a_{12}^-&=-\delta,&a_{12}^+&=\kappa B_\ast-\delta,\\
a_{21}^-&=-\left(\dfrac{\beta\zeta}{\eta}+\gamma\right),&a_{21}^+&=\varpi P_\ast-\gamma,\\
a_{22}^-&=\delta-\mu+\dfrac{\xi\zeta}{\eta},&a_{22}^+&=\dfrac{\gamma B_\ast}{P_\ast}+\varsigma P_\ast.
\end{align}
\end{subequations}
The characteristic polynomials of $\tens{A}_\pm$ are
\begin{subequations}
\begin{align}
P^\pm(\lambda)&=\lambda^4+\bigl(1+d^{-1}\bigr)c\lambda^3+\left(\dfrac{c^2}{d}-a_{11}^\pm-\dfrac{a_{22}^\pm}{d}\right)\lambda^2\nonumber\\
&\qquad-\dfrac{a_{11}^\pm+a_{22}^\pm}{d}c\lambda+\dfrac{a_{11}^\pm a_{22}^\pm-a_{12}^\pm a_{21}^\pm}{d}\label{eq:charpoly}\\
&=\lambda^4+p_3\lambda^3+p_2^\pm\lambda^2+p_1^\pm\lambda+p_0^\pm,\label{eq:charpoly2}
\end{align}
\end{subequations}
in which the first line defines the coefficients $\smash{p_0^\pm,p_1^\pm,p_2^\pm,p_3}$ in the second line. The proof now divides into two parts.
\subsubsection{Conditions for the existence of a travelling wave with $c>0$}
To analyse the necessary conditions for the existence of a travelling wave of invasion by the competitors ($c>0$), we start by observing the following:
\begin{obs}\label{obs1}
The real part of the eigenvalues of $\tens{A}_\pm$ can change sign only at $c=0$.
\end{obs}\vspace{-9pt}
\begin{proof}
Suppose to the contrary that the real part of an eigenvalue $\lambda$ of $\tens{A}_\pm$ changes sign at some $c$. Now these eigenvalues are the roots of the polynomial equation $P^\pm(\lambda)=0$, which vary continuously with $c$ since its leading coefficient is independent of $c$. Hence there is an eigenvalue $\lambda=\ic\ell$, so that $\operatorname{Re}{\lambda}=0$. If $\ell=0$, then $a_{11}^\pm a_{22}^\pm-a_{12}^\pm a_{21}^\pm=0$ from Eq.~\eqref{eq:charpoly}. This condition is independent of $c$ from Eqs.~\eqref{eq:acoeffs}, so this is a degenerate case (of zero measure) that we will not analyse further~\footnote{We neglect degenerate parameter values leading, e.g., to zero eigenvalues, tacitly throughout this paper. This is of course a mathematical assumption, but does not restrict the physical argument}. Hence $\ell\neq 0$. Since $P^\pm$ has real coefficients, its roots come in pairs of complex conjugates, so $P^\pm(-\ic\ell)=0$. Thus $\ell^4\mp\ic p_3\ell^3-\smash{p_2^\pm}\ell^2\pm\ic \smash{p_1^\pm}\ell+\smash{p_0^\pm}=0$ from Eq.~\eqref{eq:charpoly2}. Subtracting and adding these results, $\smash{\ell^4-p_2^\pm\ell^2+p_0^\pm=p_3\ell^3-p_1^\pm\ell=0}$. Since $\ell\neq 0$, the second condition implies
\begin{align}
\ell^2=\dfrac{p_1^\pm}{p_3}=-\dfrac{a_{11}^\pm+a_{22}^\pm}{1+d}.
\end{align}
As $\ell$ is real, this requires $a_{11}^\pm+a_{22}^\pm<0$. Substituting this result into the first condition,
\begin{align}
0&=\bigl(p_1^\pm\bigr)^2-p_1^\pm p_2^\pm p_3+p_0^\pm(p_3)^2\nonumber\\
&=\dfrac{c^2}{d^3}\Bigl[(1+d)\bigl(a_{11}^\pm+a_{22}^\pm\bigr)c^2-\bigl(da_{11}^\pm-a_{22}^\pm\bigr)^2\nonumber\\
&\qquad\qquad-a_{12}^\pm a_{21}^\pm(1+d)^2\Bigr].
\end{align}
Hence $c=0$ or 
\begin{align}
c^2=\dfrac{\bigl(da_{11}^\pm-a_{22}^\pm\bigr)^2+a_{12}^\pm a_{21}^\pm(1+d)^2}{(1+d)\bigl(a_{11}^\pm+a_{22}^\pm\bigr)}.
\end{align}
Since $a_{11}^\pm+a_{22}^\pm<0$, this requires $a_{12}^\pm a_{21}^\pm<0$. Finally, from Eqs.~\eqref{eq:acoeffs}, it is clear that $a_{11}^+,a_{22}^+>0$ and that $a_{12}^-a_{21}^->0$. This completes the proof: the sign of the real part of the eigenvalues can change only at $c=0$.
\end{proof}
\vspace{-5pt}
\noindent It thus suffices to calculate the eigenvalues of $\tens{A}_\pm$ for large $|c|$ to understand the behaviour of their real parts for all $c$. We now notice the following:
\begin{obs}
For $|c|\gg 1$, the eigenvalues of $\tens{A}_\pm$ have asymptotic scalings
\begin{align}
\lambda_1^\pm&\sim -c,&\lambda_2^\pm&\sim -\dfrac{c}{d},&\lambda_3^\pm&\sim\dfrac{x_1^\pm}{c},&\lambda_4^\pm&\sim\dfrac{x_2^\pm}{c},\label{eq:lambdaexp}
\end{align}
where $x_1^\pm,x_2^\pm$ are the roots of the quadratic
\begin{align}
Q(x)=x^2-\bigl(\tr{\tens{a}_\pm}\bigr)x+\bigl(\det{\tens{a}_\pm}\bigr)=0,\text{ with }\tens{a}_\pm=\begin{pmatrix}
a_{11}^\pm&a_{12}^\pm\\
a_{21}^\pm&a_{22}^\pm
\end{pmatrix}.\label{eq:QQ}
\end{align}
\end{obs}\vspace{-9pt}
\begin{proof}
We begin by seeking a root $\lambda\sim \Lambda_1c$ of $P^\pm(\lambda)=0$, where $\Lambda_1\neq 0$. From Eq.~\eqref{eq:charpoly}, $\smash{\Lambda_1^4\!+\!\bigl(1+d^{-1}\bigr)\Lambda_1^3\!+\!d^{-1}\Lambda_1^2=0}$ at leading order, so $\Lambda_1=-1$ or $\smash{\Lambda_1=-d^{-1}}$ for $\Lambda_1\neq 0$. Next, we seek a root $\smash{\lambda\sim\Lambda_{-1}c^{-1}}$. At leading order, Eq.~\eqref{eq:charpoly} yields $Q(\Lambda_{-1})=0$. We have thus identified the leading scalings of four different, and hence of all, eigenvalues of $\tens{A}_\pm$, unless $d=1$ or Eq.~\eqref{eq:QQ} has a double root, which are degenerate cases~\cite{Note2} that, again, we will not analyse.
\end{proof}
\vspace{-4pt}
\noindent Let $n_\pm=\left|\{\lambda\in\mathcal{E}(\tens{A}_\pm)\mid\operatorname{Re}{\lambda}\gtrless 0\}\right|$. From Eq.~\eqref{eq:QQ}, we compute the signs of the real parts of $x_1^\pm,x_2^\pm$,
\begin{subequations}\label{eq:xsigns}
\begin{align}
&\operatorname{Re}{x_1^\pm},\operatorname{Re}{x_2^\pm}>0&&\text{if }\tr{\tens{a}_\pm}>0,\det{\tens{a}_\pm}>0,\\
&\operatorname{Re}{x_1^\pm},\operatorname{Re}{x_2^\pm}<0&&\text{if }\tr{\tens{a}_\pm}<0,\det{\tens{a}_\pm}>0,\\
&\bigl(\operatorname{Re}{x_1^\pm}\bigr)\bigl(\operatorname{Re}{x_2^\pm}\bigr)<0&&\text{if }\det{\tens{a}_\pm}<0.
\end{align}
\end{subequations}
We now consider invasion by the competitors ($c>0$). From Eqs.~\eqref{eq:acoeffs}, $a_{11}^+,a_{22}^+>0$ as already noted earlier, so $\tr{\tens{a}_+>0}$. Moreover,
\begin{align}
\det{\tens{a}_+}=\gamma\dfrac{B_\ast^2}{P_\ast}+\delta\varsigma\dfrac{P_\ast^2}{B_\ast}+\gamma\kappa B_\ast+\delta\varpi P_\ast+(\varsigma-\kappa\varpi)B_\ast P_\ast.\label{eq:detaplus}
\end{align}
With the assumptions $\delta+(\gamma-1)\kappa>0$ or ${\gamma\varsigma+(\delta-\mu)\varpi>0}$, \textwholefigref{fig5} shows that $B_\ast<\delta/\kappa$ or $P_\ast<\gamma/\varpi$, so $\delta\varpi P_\ast>\kappa\varpi B_\ast P_\ast$ or $\gamma\kappa B_\ast>\kappa\varpi B_\ast P_\ast$. All other terms in Eq.~\eqref{eq:detaplus} are clearly positive, whence $\det{\tens{a}_+}>0$. Equations~\eqref{eq:lambdaexp} and~\eqref{eq:xsigns} thus yield
\begin{align}
n_+=2,&&n_-&=\left\{\begin{array}{cl}
2&\text{if }\tr{\tens{a}_-}>0,\det{\tens{a}_-}>0,\\
3&\text{if }\det{\tens{a}_-}<0,\\
4&\text{if }\tr{\tens{a}_-}<0,\det{\tens{a}_-}<0,
\end{array}\right.
\end{align}
for $c\gg 1$, and hence for all $c>0$ by Observation~\ref{obs1}. Equation~\eqref{eq:evA1} shows that the eigenvalues of $\tens{D}_-$ are real and of opposite signs, while, from Eq.~\eqref{eq:evB1} and for $c>0$, the real parts of the eigenvalues of $\tens{D}_+$ are negative if $\zeta>\vartheta B_\ast+\iota P_\ast$ and have opposite signs otherwise. From this, we compute $N_\pm=\left|\{\lambda\in\mathcal{E}(\tens{J}_\pm)\mid\operatorname{Re}{\lambda}\gtrless 0\}\right|$, finding
\begin{subequations}
\begin{align}
N_-&=\left\{\begin{array}{cl}
3&\text{if }\tr{\tens{a}_-}>0,\det{\tens{a}_-}>0,\\
4&\text{if }\det{\tens{a}_-}<0,\\
5&\text{if }\tr{\tens{a}_-}<0,\det{\tens{a}_-}<0,
\end{array}\right.\\
N_+&=\left\{\begin{array}{cl}
2&\text{if }\zeta>\vartheta B_\ast+\iota P_\ast,\\
3&\text{if }\zeta<\vartheta B_\ast+\iota P_\ast.
\end{array}\right.
\end{align}
\end{subequations}
This shows that, for $c>0$,
\begin{align}
N_++N_-=5\quad\text{if }\tr{\tens{a}_-}>0,\det{\tens{a}_-}>0,\zeta>\vartheta B_\ast+\iota P_\ast,\label{eq:twcondA2}
\end{align}
and $N_++N_->5$ otherwise. This generalises Eq.~\eqref{eq:nc} of the main text beyond the asymptotic regime $\varepsilon\ll 1$ and provides the necessary conditions for the existence of a travelling wave of invasion by the competitors ($c>0$).
\subsubsection{Derivation of the minimum wavespeed}
We now claim that the eigenvalues of $\tens{J}_-$ are real for all for all $c>0$ if conditions~\eqref{eq:twcondA2} hold. Since the eigenvalues of $\tens{D}_-$ in Eq.~\eqref{eq:evA1} are real, we focus on those of $\tens{A}_-$. The key observation is the following:
\begin{obs}
For $c\gg 1$, $\tens{A}_-$ has two positive and two negative real eigenvalues.
\end{obs}\vspace{-8pt}
\begin{proof}
The eigenvalues of $\tens{A}_-$ are real for $c\gg 1$. Indeed, the corrections to the leading scalings~\eqref{eq:lambdaexp} are determined by linear equations (with real coefficients), so do not introduce imaginary parts if the eigenvalues are real to leading order. Hence it suffices to show that $x_1^\pm,x_2^\pm$ are real, which requires
\begin{align}
0<(\tr{\tens{a}_-})^2-4\det{\tens{a}_-}=\bigl(a_{11}^--a_{22}^-\bigr)^2+4a_{12}^- a_{21}^-,
\end{align}
which holds true because ${a_{12}^-a_{21}^->0}$ from Eqs.~\eqref{eq:acoeffs}. Finally, Eqs.~\eqref{eq:lambdaexp} now imply that $\lambda_1^-,\lambda_2^-<0$ for $c>0$. Moreover, conditions~\eqref{eq:twcondA2} include $\tr{\tens{a}_-},\det{\tens{a}_-}>0$. Hence $\lambda_3^-,\lambda_4^->0$ for $c>0$ from Eqs.~\eqref{eq:lambdaexp} and~\eqref{eq:xsigns}.
\end{proof}\vspace{-4pt}
\noindent By Observation~\ref{obs1}, it follows that $\tens{A_-}$ has two eigenvalues with positive real parts and two eigenvalues with negative real parts for all $c>0$. As noted earlier, the eigenvalues of $\tens{A}_-$ vary continuously with $c$. Hence, if they complexify at some $c>0$, one of these pairs of pairs of eigenvalues must merge. Now, from Eq.~\eqref{eq:charpoly}, the characteristic polynomial of $\tens{A}_-$ can be rearranged into a quadratic equation for $c$,
\begin{align}
&\left(\lambda^4+\dfrac{\lambda^2}{d}\right)c^2+\left[\bigl(1+d^{-1}\bigr)\lambda^3+\dfrac{a_{11}^-+a_{22}^-}{d}\lambda\right]c\nonumber\\
&\quad+\dfrac{a_{11}^- a_{22}^--a_{12}^- a_{21}^-}{d}-\left(a_{11}^-+\dfrac{a_{22}^-}{d}\right)\lambda^2=0.
\end{align}
If two eigenvalues merge, then the discriminant of this polynomial must vanish because there are no other eigenvalues with $\lambda>0$ or $\lambda<0$. This discriminant simplifies to
\begin{align}
\bigl[a_{11}^--a_{22}^--(d-1)\lambda^2\bigr]^2+4a_{12}^- a_{21}^->0,
\end{align}
because $a_{12}^- a_{21}^->0$ as noted earlier. This is a contradiction. Thus the eigenvalues of $\tens{A}_-$ and hence $\tens{J}_-$ are real for all $c>0$.

Accordingly, the only eigenvalues that might complexify are those of $\tens{J}_+$. As noted in the analysis of the asymptotic limit ${\varepsilon\ll 1}$ in the main text, the block matrix structure of $\tens{J}_+$ in Eq.~\eqref{eq:Jplus} shows that the eigenvectors of $\tens{A}_+$ are orthogonal to the $(A,T)$ plane, so do not affect the dynamics in this plane. We can therefore, again, focus on the eigenvalues of $\tens{D}_+$, given by Eq.~\eqref{eq:evA2}. They can complexify since, from conditions~\eqref{eq:twcondA2}, ${\zeta>\vartheta B_\ast+\iota P_\ast}$. With this assumption, these eigenvalues are stable. As in the main text~\wholefigref{fig2}, this implies a minimum wave speed, $c\geqslant c_\text{min}$, with
\begin{align}
c_\text{min}&=2\sqrt{D(\zeta-\vartheta B_\ast-\iota P_\ast)},
\end{align}
as obtained in the main text for $\varepsilon\ll 1$. This completes the proof of Proposition~\ref{prop2}.\hfill$\Box$
\section{\uppercase{Travelling two-species waves without phenotypic switching}}\label{appB}
In this Appendix, we derive algebraic conditions on travelling-wave solutions of the two-species Lotka--Volterra model
\begin{subequations}\label{eq:lv2}
\begin{align}
\dfrac{\partial B}{\partial t}&=B(\omega - \alpha A - \chi B) + D_B\frac{\partial^2 B}{\partial x^2},\\
\dfrac{\partial A}{\partial t}&= A(\zeta - \eta A - \vartheta B) + D_A\frac{\partial^2 A}{\partial x^2}.
\end{align}
\end{subequations}
We have omitted to rescale populations, time, and space to set, e.g., $\omega=\chi=D_B=1$ in Eqs.~\eqref{eq:lv2}, because we need this ``dimensional'' form of the equations to define an average of Eqs.~\eqref{eq:lv} without phenotypic variation in the main text.

In addition to a possible steady state of coexistence, Eqs.~\eqref{eq:lv2} admit two steady states of one species only,
\begin{align}
&\mathcal{A}=(0,A_\ast),&&\mathcal{B}=(B_\ast,0),
\end{align}
where $A_\ast=\zeta/\eta$ and $B_\ast=\omega/\chi$. In what follows, we will determine algebraic conditions for the existence of a travelling wave by which one of $\mathcal{A},\mathcal{B}$ invades the other. These results are probably folklore (although we could not find an exact reference), but their derivation lays out, in a pedagogical manner, the ideas underpinning the more complex calculations in the main text.
\subsection*{Travelling-wave solutions}
As in the main text, we introduce the travelling-wave coordinate $z=x-ct$ to recast Eqs.~\eqref{eq:lv2} into
\begin{subequations}\label{eq:tweq2}
\begin{align}
B'&=R,& D_BR'&=-cR-B(\omega-\alpha A-\chi B),\\
A'&=T,& D_AT'&=-cT-A(\zeta-\eta A-\vartheta B).
\end{align}
\end{subequations}
We let $\vec{X}=(B,R,A,T)$, and seek solutions such that $\vec{X}\to\vec{X_\pm}$ as $z\to\pm\infty$, where $\vec{X_-}=(0,0,A_\ast,0)$, $\vec{X_+}=(B_\ast,0,0,0)$. Near these equilibria, Eqs.~\eqref{eq:tweq2} become $\vec{X'}=\tens{J}_\pm\cdot\vec{X}$, with
\begin{subequations}\label{eq:JsAppB}
\begin{align}
\tens{J}_-&=\begin{pmatrix}
0&1&0&0\\
\dfrac{\alpha\zeta-\omega\eta}{D_B\eta}&-\dfrac{c}{D_B}&0&0\\
0&0&0&1\\
\dfrac{\zeta\vartheta}{D_A\eta}&0&\dfrac{\zeta}{D_A}&-\dfrac{c}{D_A}
\end{pmatrix},\\
\tens{J}_+&=\begin{pmatrix}
0&1&0&0\\
\dfrac{\omega}{D_B}&-\dfrac{c}{D_B}&\dfrac{\alpha\omega}{D_B\chi}&0\\
0&0&0&1\\
0&0&\dfrac{\vartheta\omega-\zeta\chi}{D_A\chi}&-\dfrac{c}{D_A}
\end{pmatrix}.
\end{align}
\end{subequations}
The eigenvalues of $\tens{J}_-$ are
\begin{subequations}\label{eq:ev2}
\begin{align}
\lambda_-&=\dfrac{1}{2D_A}\left(-c\pm\sqrt{c^2+4D_A\zeta}\right),\label{eq:eva}\\
\lambda_-&=\dfrac{1}{2D_B}\left[-c\pm\sqrt{c^2+4D_B\left(\dfrac{\alpha\zeta}{\eta}-\omega\right)}\right],\label{eq:evb}
\end{align}
\end{subequations}
while those of $\tens{J}_+$ are
\begin{subequations}\label{eq:ev2B}
\begin{align}
\lambda_+&=\dfrac{1}{2D_B}\left(-c\pm\sqrt{c^2+4D_B\zeta}\right),\label{eq:evc}\\
\lambda_+&=\dfrac{1}{2D_A}\left[-c\pm\sqrt{c^2+4D_A\left(\dfrac{\vartheta\omega}{\chi}-\zeta\right)}\right].\label{eq:evd}
\end{align}
\end{subequations}
\subsubsection{Necessary conditions for the existence of a travelling wave}
Equations~\eqref{eq:tweq2} define a fourth-order system of ordinary differential equations. We can therefore impose 4 boundary conditions. Since the problem is invariant under translations in~$z$, one of these boundary conditions must remove this freedom. The remaining 3 boundary conditions must eliminate the growing modes as $z\to\pm\infty$. These are associated to eigenvalues~$\lambda_\pm$ of $\tens{J}_\pm$ such that $\operatorname{Re}{\lambda_\pm}\gtrless 0$. A necessary condition for the existence of a travelling wave is therefore
\begin{align}
N_++N_-=3,\quad\text{where }N_\pm=\left|\{\lambda_\pm\in\mathcal{E}(\tens{J}_\pm)\mid\operatorname{Re}{\lambda_\pm}\gtrless 0\}\right|.
\end{align}
From Eqs.~\eqref{eq:ev2} and~\eqref{eq:ev2B}, we compute
\begin{subequations}
\begin{align}
N_-&=\left\{\begin{array}{cl}
1&\text{if $\zeta<\dfrac{\omega\eta}{\alpha}$ and $c<0$}, \\[3mm]
2&\text{if $\zeta>\dfrac{\omega\eta}{\alpha}$},\\[3mm]
3&\text{if $\zeta<\dfrac{\omega\eta}{\alpha}$ and $c>0$},
\end{array}\right.\\
N_+&=\left\{\begin{array}{cl}
1&\text{if $\zeta>\dfrac{\omega\vartheta}{\chi}$ and $c>0$}, \\[3mm]
2&\text{if $\zeta<\dfrac{\omega\vartheta}{\chi}$},\\[3mm]
3&\text{if $\zeta>\dfrac{\omega\vartheta}{\chi}$ and $c<0$},
\end{array}\right.
\end{align}
whence
\begin{align}
N_++N_-=3\quad \text{if}\left\{\begin{array}{l}
\zeta<\min{\left\{\dfrac{\omega\vartheta}{\chi},\dfrac{\omega\eta}{\alpha}\right\}}\text{ and }c<0\text{ or} \\[3mm]
\zeta>\max{\left\{\dfrac{\omega\vartheta}{\chi},\dfrac{\omega\eta}{\alpha}\right\}}\text{ and }c>0,
\end{array}\right.\label{eq:ncs}
\end{align}
\end{subequations}
and $N_++N_->3$ otherwise. In this way, this calculation provides necessary conditions for the existence of a travelling wave of invasion: the competitors can invade ($c>0$) only if their growth rate is sufficiently large, $\zeta>\max{\{\omega\vartheta/\chi,\omega\eta/\alpha\}}$. Conversely, the competitors can be invaded ($c<0$) only if their growth rate is sufficiently small, $\zeta<\min{\{\omega\vartheta/\chi,\omega\eta/\alpha\}}$.

If neither of these conditions holds, i.e., if ${\vartheta/\chi\!<\!\zeta/\omega\!<\!\eta/\alpha}$ or $\eta/\alpha<\zeta/\omega<\vartheta/\chi$, then we might expect a double travelling wave, through which a steady state $\mathcal{C}$ of coexistence of both species to invade both $\mathcal{A}$ and $\mathcal{B}$. We will not analyse this possibility in detail, but we note that
\begin{align}
\mathcal{C}=\left(\dfrac{\vartheta\omega-\zeta\chi}{\alpha\vartheta-\eta\chi},\dfrac{\alpha\zeta-\eta\omega}{\alpha\vartheta-\eta\chi}\right).
\end{align}
Hence $\mathcal{C}$ is feasible (i.e., non-negative) if and only if $\vartheta\omega-\zeta\chi$, $\alpha\vartheta-\eta\chi$, $\alpha\zeta-\eta\omega$ all have the same sign. We notice that $\alpha(\vartheta\omega-\zeta\chi)+\chi(\alpha\zeta-\eta\omega)=\omega(\alpha\vartheta-\eta\chi)$, so this is if and only if $\vartheta\omega-\zeta\chi$ and $\alpha\zeta-\eta\omega$ have the same sign. Thus coexistence is feasible if and only if if ${\vartheta/\chi<\zeta/\omega<\eta/\alpha}$ or $\eta/\alpha<\zeta/\omega<\vartheta/\chi$, i.e., if and only if there is no travelling wave by which one species invades the other.

\subsubsection{Minimum wavespeed of travelling waves}
Finally, we obtain a lower bound $|c|>|c_{\min}|$ on the speed of these travelling waves by noting that the non-growing eigenmodes (that are not eliminated by the $N_++N_-=3$ boundary conditions imposed) that dominate the dynamics as $z\to\pm\infty$ must be real. Indeed, if they were complex, the dynamics would describe a spiral near $\vec{X_+}$ or $\vec{X_-}$, contradicting $A(z),B(z)\geq 0$. 

Only eigenvalues~\eqref{eq:evb} and~\eqref{eq:evd} can be complex. In a travelling wave with $c>0$ (invasion by the competitors), we have $\zeta>\omega\eta/\alpha$ from Eq.~\eqref{eq:ncs}, so eigenvalues~\eqref{eq:evb} are real, but we also have $\zeta>\omega\vartheta/\chi$, so eigenvalues~\eqref{eq:evd} are real if and only if $c^2\geqslant 4D_A(\zeta-\vartheta\omega/\chi)$. Similarly, in a travelling wave with $c<0$ (invasion of the competitors), we must have $c^2\geqslant 4D_B(\omega-\alpha\zeta/\eta)$. This suggests that $|c|\geqslant |c_{\min}|$, where
\begin{align}
c_{\min}=\left\{\begin{array}{cl}
-2\sqrt{D_B\left(\omega-\dfrac{\alpha\zeta}{\eta}\right)}&\text{if $c<0$},\\
2\sqrt{D_A\left(\zeta-\dfrac{\vartheta\omega}{\chi}\right)}&\text{if $c>0$}.
\end{array}\right.
\end{align}
To complete the argument, we need to show that these eigenmodes dominate the dynamics. This follows from the diagonal block structure of the Jacobians in Eqs.~\eqref{eq:JsAppB}: the eigenmodes corresponding to Eqs.~\eqref{eq:eva} and~\eqref{eq:evc} are orthogonal to the $(B,R)$ and $(A,T)$ planes, respectively, so the dynamics in those planes are dominated by the eigenmodes corresponding to Eqs.~\eqref{eq:evb} and~\eqref{eq:evd}, the complexification of which thus indeed sets the minimum wavespeed.
\section{\uppercase{Numerical methods}}\label{appC}
In this final Appendix, we describe the implementation~\cite{Note1} of our numerical study of the full model~\eqref{eq:lv} and the average model~\eqref{eq:lvav}.
\subsection{Numerical integration}
We integrate the partial differential equations~\eqref{eq:lv} and~\eqref{eq:lvav} in \textsc{Matlab} (The MathWorks, Inc.) using the \texttt{pdepe} function. We take the initial conditions for Eqs.~\eqref{eq:lv} to be
\begin{subequations}\label{eq:ics}
\begin{align}
B(x,0)=\dfrac{B_\ast}{2}(1+\tanh{kx}),\\
P(x,0)=\dfrac{P_\ast}{2}(1+\tanh{kx}),\\
A(x,0)=\dfrac{A_\ast}{2}(1-\tanh{kx}),
\end{align}
\end{subequations}
which defines the steepness parameter $k$. To compute these initial conditions and the parameters of the averaged model~\eqref{eq:lvav}, we determine $B_\ast,P_\ast$ by solving the steady-state equations~\eqref{eq:BPeq} using the \texttt{fminsearch} function, starting from the initial guess $(B_\ast,P_\ast)=(1,\gamma/\delta)$ that corresponds to their leading-order solution for $\varepsilon\ll 1$. Equations~\eqref{eq:ccpop} show that Eqs.~\eqref{eq:ics} are consistent with
\begin{align}
\bar{B}(x,0)&=\dfrac{\bar{B}_\ast}{2}(1+\tanh{kx}),&\bar{A}(x,0)&=\dfrac{\bar{A}_\ast}{2}(1-\tanh{kx}).
\end{align}
To avoid the wave front leaving the spatial domain, we change to a frame moving at velocity $C$ by recasting Eqs.~\eqref{eq:lv} in terms of the coordinate $X=x-Ct$,
\begin{subequations}\label{eq:lv2m}
\begin{align}
\dfrac{\partial B}{\partial t}&=B(1\!-\!\alpha A\!-\!B\!-\!\kappa P )\!-\!\beta AB\!-\!\gamma B\!+\!\delta P+\frac{\partial^2 B}{\partial X^2}+C\dfrac{\partial B}{\partial X},\\
\dfrac{\partial P}{\partial t}&=P(\mu\!-\!\xi A\!-\!\varpi B\!-\!\varsigma P )\!+\!\beta AB\!+\!\gamma B\!-\!\delta P\!+\! d\frac{\partial^2 P}{\partial X^2}\!+\!C\dfrac{\partial P}{\partial X},\\
\dfrac{\partial A}{\partial t}&= A(\zeta - \eta A - \vartheta B - \iota P ) + D\frac{\partial^2 A}{\partial X^2}+C\dfrac{\partial A}{\partial X}.
\end{align}
\end{subequations}
We initially take $C=c_\text{min}$, and integrate the equations for ${X\in[-M,M]}$, where $M\gg 1$, and ${t\in[0,T]}$ using Neumann boundary conditions at $X=\pm M$. We then determine the position $x_\ast(T)$ of the wave front, where $A(x_\ast(T),T)=A_\ast/2$, and update $C_\text{new} =C_\text{old}+x_\ast/T$. We then move the wave front back to $X=0$, and iterate until ${|C_\text{new}-C_\text{old}|/|C_\text{old}|<\text{tol}=10^{-3}}$. Convergence determines $c_\ast=C$~\cite{Note1}. We determine $\bar{c}_\ast$ analogously by iterative integration of Eqs.~\eqref{eq:lvav}.
\subsection{Sampling of random systems}
To sample random systems, we sample the model parameters independently from uniform distributions. When sampling parameters obeying the persister scalings~\eqref{eq:scalings}, parameters that are of order $O(\varepsilon)$ are sampled from the $\mathcal{U}[0,\varepsilon]$ distribution, while other parameters are sampled from the $\mathcal{U}[1/f,f]$ distribution for some $f>1$. This allows sampling the parameters of order $O(1)$ relative to those parameters eliminated by the nondimensionalisation of Eqs.~\eqref{eq:lv}. When sampling general systems, all parameters are sampled from the $\mathcal{U}[1/f,f]$ distribution; we will refer to this sampling as $\varepsilon=1$ below.
\subsection{Travelling-wave conditions}
For each random system sampled, we check whether the necessary conditions for the existence of a travelling wave of invasion are satisfied in the full or averaged models, and, if they are, compute the minimum wave velocities $c_\text{min}$ or $\bar{c}_\text{min}$. (We only integrate the dynamics for random systems satisfying these conditions.)
\subsubsection{Necessary conditions for travelling waves of invasion}
\citeappendix{appB} shows that the necessary conditions for a travelling wave of invasion in the averaged model~\eqref{eq:lvav} are
\begin{align}
&\dfrac{\bar{\zeta}}{\bar\omega}<\min{\left\{\dfrac{\bar\vartheta}{\bar\chi},\dfrac{\bar\eta}{\bar\alpha}\right\}}\;\;(c<0),&&\dfrac{\bar{\zeta}}{\bar\omega}>\max{\left\{\dfrac{\bar\vartheta}{\bar\chi},\dfrac{\bar\eta}{\bar\alpha}\right\}}\;\;(c>0).
\end{align}
For the full model~\eqref{eq:lv}, we have derived the necessary conditions for the existence of a travelling wave with $c>0$ in the proof of Proposition~\ref{prop2} in \citeappendix{appA}. However, the conditions of Proposition~\ref{prop2} may fail to hold for general random systems, but, noting that $\tr{\tens{a}_+}>0$ independently of these conditions, the argument generalises straightforwardly: letting $Z=\zeta-\vartheta B_\ast-\iota P_\ast$, we obtain the necessary conditions
\begin{subequations}
\begin{align}
\left.\begin{array}{c}
\tr{\tens{a}_-}>0\land\det{\tens{a}_-}>0\land Z\det{\tens{a}_+}>0\\
\lor\\
\det{\tens{a}_-}<0\land\det{\tens{a}_+}<0\land Z>0
\end{array}\right\}(c>0),\\
\left.\begin{array}{c}
\tr{\tens{a}_-}<0\land\det{\tens{a}_-}>0\land Z\det{\tens{a}_+}>0\\
\lor\\
\det{\tens{a}_-}<0\land\det{\tens{a}_+}>0\land Z<0
\end{array}\right\}(c<0).
\end{align}
\end{subequations}
(For the avoidance of doubt, we add that the logical $\land$ takes precedence over the logical $\lor$ in the above conditions.)
\subsubsection{Calculation of the minimum wave speed}
The minimum wavespeed in the averaged model~\eqref{eq:lvav} is given by Eq.~\eqref{eq:cminav}. To calculate the minimum wavespeed in the full model~\eqref{eq:lv}, including in cases not covered by the conditions of Proposition~\ref{prop2}, we extend the argument of the main text and of the proof of this proposition in \citeappendix{appA}: in the notation of \citeappendix{appA}, the eigenvectors of $\tens{D}_-$ and $\tens{A}_+$ are orthogonal to the $(B,R,P,S)$ and $(A,T)$ planes, respectively, so the minimum wavespeed is determined by requiring that the (non-growing) eigenvalues of $\tens{A}_-$ and $\tens{D}_+$ that dominate the dynamics as $z\to-\infty$ and $z\to\infty$, respectively, be real. The realness of the eigenvalues of $\tens{A}_-,\tens{D}_+$ changes at values of $c$ where two eigenvalues merge, i.e., at the roots of the discriminants of the characteristic polynomials of $\tens{A}_-$ or~$\tens{D}_+$. [In fact, for $c<0$, Eq.~\eqref{eq:evA1} shows that the eigenvalues of $\tens{D}_+$ are growing ($\operatorname{Re}{\lambda_+}>0$) if they are complex, so the minimum wave speed is determined by $\tens{A}_-$ only if $c<0$.] We therefore calculate these roots and determine the realness of the eigenvalues (and hence the minimum wavespeed) by computing them at the midpoint of the intervals between roots for numerical stability~\cite{Note1}.
\subsection{Additional discussion of the numerical results}
In Figs.~\ref{fig3} and~\ref{fig4} of the main text, we report wave speeds for those randomly sampled systems for which the necessary conditions for the existence of a travelling wave are satisfied in both the full model~\eqref{eq:lv} and the averaged model~\eqref{eq:lvav} and for which the computations of both wave speeds converge. In this final subsection, we discuss the remaining systems.
\subsubsection{Robustness of the numerical method}
In random systems satisfying these necessary conditions for both models, numerical integrations of one or both models failed in less than $1\%$ of cases. Most of these failures resulted from the wave speed not converging within $N_\text{max}=50$ iterations. This low failure rate confirms the robustness of our numerical approach.
\subsubsection{Systems for which travelling waves exist in one model only}
We conclude by discussing those randomly sampled systems in which the necessary conditions are satisfied in one, but not both models. The proportion of such systems is negligible if the persister scalings~\eqref{eq:scalings} are satisfied, but becomes more significant for general random systems: we compute $0.0036$ ($[0.0019,0.0069]$) for $\varepsilon=0.1$, but $0.080$ ($[0.069,0.091]$) for $\varepsilon=1$~\cite{*[{The confidence intervals are 95\% Wilson intervals: see, e.g., }][] brown01}. In the latter case, we find that it is somewhat more likely that the travelling wave conditions are only satisfied in the full model (one-sided binomial test, $p<0.1$). In other words, phenotypic switching can favour invasion of one species by the other.
\bibliography{main}
\end{document}